\documentclass{article}
\usepackage[english]{babel}
\usepackage{geometry,amsmath,amssymb,enumerate,latexsym}
\usepackage{hyperref}
\newcommand{\assign}{:=}
\newcommand{\tmaffiliation}[1]{\\ #1}
\newcommand{\tmop}[1]{\ensuremath{\operatorname{#1}}}
\newenvironment{enumerateroman}{\begin{enumerate}[i.] }{\end{enumerate}}
\newenvironment{proof}{\noindent\textbf{Proof\ }}{\hspace*{\fill}$\Box$\medskip}
\newtheorem{corollary}{Corollary}
\newtheorem{definition}{Definition}
\newtheorem{lemma}{Lemma}
\newtheorem{theorem}{Theorem}

\begin{document}

\title{Nonuniform QCPH Collapse Implies QCPH Collapse}

\author{
  Jeremy Ahrens Huang
  \tmaffiliation{Department of Computer Science and Engineering, Pennsylvania
  State University\\
  Email: jeremyah@psu.edu}
}

\maketitle

\begin{abstract}
    Despite the importance of the non-uniform Polynomial-Time Hierarchy (PH/poly) in classical complexity understanding the collapse conditions of the Polynomial-Time Hierarchy (PH), a quantum equivalent of PH/poly has yet to be studied in the literature. We introduce the non-uniform computational Quantum Polynomial-Time Hierarchy (QCPH/mpoly), the quantum equivalent of PH/poly, and show that it collapses if and only if QCPH, the quantum equivalent of PH introduced by Gharibian et al. (comput. complex. 2022), also collapses. We also show that QCPH/mpoly collapses if coQCMA is contained in QCMA/mpoly. These results are analogous to those of Yap (TCS 1983) commonly used to invoke the collapse of PH in classical complexity. QCPH/mpoly is analogous to QCPH with non-uniform quantum verifier circuits. 
\end{abstract}

\section{Introduction}
The hypothesis that the the Polynomial-Time Hierarchy (PH) does not collapse is one of the most popular and most relied-upon in classical complexity theory as it is an even weaker assumption than P $\neq$ NP! A common way that PH collapse is invoked is through statements about the non-uniform Polynomial-Time Hierarchy $\tmop{PH}/\tmop{poly}$ such as ``small circuits for $\tmop{coNP}$ implies $\tmop{coNP} \subseteq \tmop{NP}/\tmop{poly}$'' due to two foundational results of \cite{Yap83}: first that $\Pi_i \subseteq \Sigma_i/\tmop{poly}$ implies the collapse of $\tmop{PH}/\tmop{poly}$, and second that $\tmop{PH}/\tmop{poly}$ collapses if and only if PH collapses. Surprisingly, the quantum analog of PH, QCPH, has only been defined in the last decade \cite{GSS+22}\footnote{Other quantum hierarchies were defined earlier, but they were less analogous to PH. QCPH generalizes the relationship between BQP and MQA, just as PH generalizes the relationship between the analogous classes in the classical deterministic setting P and NP.}, and a quantum analogue of $\tmop{PH}/\tmop{poly}$ has yet to be defined. In this paper we introduce the non-uniform (computational) Quantum Polynomial-Time Hierarchy $\tmop{QCPH}/\tmop{poly}$ and establish three collapse results similar to the results of \cite{Yap83} for $\tmop{PH}/\tmop{poly}$.

First we establish a collapse theorem for $\tmop{QCPH}/\tmop{poly}$. Though this result is expected, it is not a given! No collapse theorem for $\tmop{QCPH}$ was known when it was first defined; a collapse theorem was only recently given by \cite{AGKR24}.

\begin{theorem} \label{thm:advice-collapse}
  For $i \in \mathbb{Z}^+$,
  
  $\tmop{QC} \Sigma_i / \tmop{mpoly} = \tmop{QC} \Pi_i / \tmop{mpoly}$
  $\Rightarrow$ $\tmop{QC} \Sigma_i / \tmop{mpoly} = \tmop{QCPH} /
  \tmop{mpoly} = \tmop{QC} \Pi_i / \tmop{mpoly}$.
\end{theorem}

Second we show that $\tmop{QCPH}/\tmop{poly}$ also collapses under a weaker condition.

\begin{theorem} \label{thm:non-advice-inclusion}
  For $i \in \mathbb{Z}^+$,
  
  $\tmop{QC} \Pi_i \subseteq \tmop{QC} \Sigma_i / \tmop{mpoly}$
  $\Rightarrow$ $\tmop{QC} \Sigma_i / \tmop{mpoly} = \tmop{QCPH} /
  \tmop{mpoly} = \tmop{QC} \Pi_i / \tmop{mpoly}$.
\end{theorem}

Last but not least, we show that either both $\tmop{QCPH}$ and $\tmop{QCPH}/\tmop{poly}$ collapse, or neither collapses. 

\begin{theorem}
  \label{thm:collapse-collapse}For $i \in \mathbb{Z}^+$, $\tmop{QC} \Sigma_i /
  \tmop{mpoly} = \tmop{QC} \Pi_i / \tmop{mpoly}$ $\Rightarrow$ $\tmop{QC}
  \Sigma_{i + 2} = \tmop{QC} \Pi_{i + 2}$.
\end{theorem}

Taken all together, these results give broader and more useful conditions for QCPH collapse such as $\tmop{coMQA} \subseteq \tmop{MQA}/\tmop{poly}$ (the quantum equivalent of $\tmop{coNP} \subseteq \tmop{NP}/\tmop{poly}$), bringing the quantum complexity toolbox closer to parity with classical complexity. They also strengthen and generalize the ``Karp-Lipton for QCPH'' result from \cite{AGKR24}: \[\tmop{QC\Sigma}_1 \subseteq \tmop{QC\Pi}_0 / \tmop{mpoly} \Rightarrow \tmop{QCPH} = \tmop{QC\Pi}_2 = \tmop{QC\Sigma}_2.\]

\subsection{Prior work}

The Polynomial-Time Hierarchy (PH) was first introduced by \cite{Sto76}. It generalized the relationship between P, the set of languages decidable by a deterministic classical computer in polynomial time, and NP, the set of languages decidable by a deterministic classical computer in polynomial time when given a certificate that can be generated by a deterministic classical computer in polynomial time. PH can be defined in the following way: let $\Sigma_i$ be the set of languages such that $L\in\Sigma_i$ if and only if 
\begin{eqnarray*}
    x \in L & \Leftrightarrow & \exists y_1 \forall y_2 \exists y_3
    \ldots Q_i y_i [C (x, y_1, y_2, \ldots, y_i) = 1]
\end{eqnarray*}
for some uniformly-generate classical circuit family $C$, and let $\Pi_i$ be defined similarly except with the quantifier order swapped. Then \[\tmop{PH} = \bigcup_{i>0}\Sigma_i=\bigcup_{i>0}\Pi_i.\] This provided an alternate way to study the relationship between P and NP and a weaker assumption than P $\neq$ NP to base complexity conjectures on. \cite{KL80} formalized non-uniform classes and showed that $\tmop{\Sigma}_1 \subseteq \tmop{\Pi}_0 / \tmop{poly}$ implied the collapse of PH, introducing the use of recursive structures to prove collapse results. \cite{Yap83} introduced PH/poly, showed that its collapse was tied to the collapse of PH, and improved and generalized the condition of the \cite{KL80} result to $\tmop{\Sigma}_i \subseteq \tmop{\Pi}_i / \tmop{poly}$. PH/poly can be thought of as PH with non-uniform circuits or with uniform circuits that receive an advice string with every input that depends only on the input length. 

\cite{GSS+22} were the first to introduce the quantum equivalent of PH, which they call QCPH. It generalizes the relationship between BQP, the set of languages decidable by a randomized quantum computer in polynomial time, and MQA (also known as QCMA), the set of languages decidable by a randomized quantum computer in polynomial time when given a certificate that can be generated by a randomized quantum computer in polynomial time. QCPH can be thought of as PH but with a bounded-error quantum verifier circuit instead of a deterministic classical circuit. Unfortunately, they were not able to show a condition under which QCPH would collapse. \cite{AGKR24} showed that QCPH collapses if $\tmop{QC\Pi_i}=\tmop{QC\Sigma_i}$ using better techniques for working with quantified circuits and \cite{FGN23} showed that QCPH collapses by giving an alternate definition for it using variants of quantifier operators on classes that work better for promise problems. \cite{AGKR24} also showed a \cite{KL80}-like result for MQA and $\tmop{BQP}/\tmop{mpoly}$. 

Before \cite{GSS+22}, other quantum hierarchies were defined by \cite{Yam02, LG17}.

\subsection{Acknowledgments}
I want to thank Avantika Agarwal for a helpful discussion of the QCPH collapse theorem. 

\paragraph{AI Statement} No LLM tools were used in the creation of this paper.

\section{Preliminaries}

We start with a definition of QCPH.

\begin{definition}
  $\tmop{QC} \Sigma_i$: Let $S = (S_{\tmop{yes}}, S_{\tmop{no}})$ be a promise
  problem. $S \in \tmop{QC} \Sigma_i$ iff there is a uniform family of
  poly-time quantum circuits $V$ such that
  \begin{eqnarray*}
    x \in S_{\tmop{yes}} & \Rightarrow & \exists y_1 \forall y_2 \exists y_3
    \ldots Q_i y_i \Pr [V (x, y_1, y_2, \ldots, y_i) = 1] \geqslant 2 / 3\\
    x \in S_{\tmop{no}} & \Rightarrow & \forall y_1 \exists y_2 \forall y_3
    \ldots \overline{Q}_i y_i \Pr [V (x, y_1, y_2, \ldots, y_i) = 1] \leqslant
    1 / 3
  \end{eqnarray*}
  where the $y_i$ are quantum states with $O (\tmop{poly} (| x |))$
  descriptive complexity for uniform quantum circuits.
\end{definition}

Note that MQA = $\tmop{QC}\Sigma_1$. Also note that the quantum states $y_i$ can be replaced by classical bit strings describing quantum circuits that generate the $y_i$ and $V$ can use those descriptions to generate $y_i$ itself. 

\begin{definition}
  $\tmop{QC} \Pi_i = \tmop{coQC} \Sigma_i$.
\end{definition}

Note that taking the complement of a promise language simply swaps the yes and no sets, leaving the set of indeterminate strings unchagned. 

\begin{definition}
  (Computational Quantum Polynomial Hierarchy) $\tmop{QCPH} = \bigcup_{i \in
  \mathbb{N}} \tmop{QC} \Sigma_i = \bigcup_{i \in \mathbb{N}} \tmop{QC}
  \Pi_i$.
\end{definition}

Despite being unresolved for half a decade, we now have a collapse theorem for QCPH. 

\begin{theorem}[\cite{AGKR24}]
  \label{thm:qcph-collapse} $\tmop{QC} \Sigma_i = \tmop{QC}
  \Pi_i$ implies $\tmop{QCPH} = \tmop{QC} \Sigma_i = \tmop{QC} \Pi_i$.
\end{theorem}

Now we are ready to define QCPH/mpoly. Naturally, we define it as QCPH but with non-uniform verifiers. We phrase the formal definition below in terms of uniform verifiers with advice because explicit functions on the advice string will be used in the proofs. Note that non-uniform classes with bounded error are named ``mpoly'' because ``poly'' has a specific definition in \cite{KL80} that causes an ``$\exists$BPP vs MA''-like problem with bounded-error classes\footnote{See \href{https://complexityzoo.net/Complexity_Zoo:B\#bqppoly}{here} for a discussion of the distinction between Merlin-like and Karp-Lipton-like advice.}.

\begin{definition}
  $\tmop{QC} \Sigma_i / \tmop{mpoly}$: Let $S = (S_{\tmop{yes}},
  S_{\tmop{no}})$ be a promise problem.
  
  $S \in \tmop{QC} \Sigma_i / \tmop{mpoly}$ iff there is an advice function $a
  : \mathbb{N} \mapsto \{ 0, 1 \}^{\ast}$ s.t. $| a (n) | \leqslant p (n)$ and
  a uniform family of poly-time quantum circuits $V$ such that
  \begin{eqnarray*}
    x \in S_{\tmop{yes}} & \Rightarrow & \exists y_1 \forall y_2 \exists y_3
    \ldots Q_i y_i \Pr [V (a_{| x |}, x, y_1, y_2, \ldots, y_i) = 1] \geqslant
    2 / 3\\
    x \in S_{\tmop{no}} & \Rightarrow & \forall y_1 \exists y_2 \forall y_3
    \ldots \overline{Q}_i y_i \Pr [V (a_{| x |}, x, y_1, y_2, \ldots, y_i) =
    1] \leqslant 1 / 3,
  \end{eqnarray*}
  where $| y_i | \leqslant q (| x |)$ for some polynomials $p, q$.
\end{definition}

\begin{definition}
  $\tmop{QC} \Pi_i / \tmop{mpoly} = \tmop{co} (\tmop{QC} \Sigma_i /
  \tmop{mpoly})$
\end{definition}

\begin{definition}
  (Nonuniform Computational Quantum Polynomial Hierarchy)
  
  $\tmop{QCPH} / \tmop{mpoly} = \bigcup_{i \in \mathbb{N}} \tmop{QC} \Sigma_i
  / \tmop{mpoly} = \bigcup_{i \in \mathbb{N}} \tmop{QC} \Pi_i / \tmop{mpoly}$
\end{definition}

We can also call this class $\tmop{QCPH}$ with deterministic Merlin-like
advice.

\section{Collapse theorems}

The following lemma will be useful for our key corollary. \cite{Yap83} gives a version for total decision problems, but we need one for promise problems.

\begin{lemma}
  \label{co-mpoly-commute}$(\tmop{co} {\textendash}K) / \tmop{mpoly} =
  \tmop{co} {\textendash} (K / \tmop{mpoly})$
\end{lemma}

\begin{proof}
  Let $L \in \tmop{co} {\textendash} (K / \tmop{mpoly})$.
  
  Since $L$ may be a promise problem, we have $L = (L_{\tmop{yes}},
  L_{\tmop{no}})$ and $\tmop{co} {\textendash}L = ((\tmop{co}
  {\textendash}L)_{\tmop{yes}}, (\tmop{co} {\textendash}L)_{\tmop{no}})$ so by
  definition $(L_{\tmop{yes}}, L_{\tmop{no}}) = ((\tmop{co}
  {\textendash}L)_{\tmop{no}}, (\tmop{co} {\textendash}L)_{\tmop{yes}})$.
  Since $\tmop{co} {\textendash}L \in K / \tmop{mpoly}$, $\tmop{co}
  {\textendash}L = L' / a$ for some $L' \in K$ and poly-size advice function
  $a_{| x |}$, so $x \in (\tmop{co} {\textendash}L)_{\tmop{yes}}
  \Leftrightarrow (x, a_{| x |}) \in L'_{\tmop{yes}}$ and $x \in (\tmop{co}
  {\textendash}L)_{\tmop{no}} \Leftrightarrow (x, a_{| x |}) \in
  L'_{\tmop{no}}$. Then
  \begin{eqnarray*}
    x \in L_{\tmop{yes}} \Leftrightarrow x \in (\tmop{co}
    {\textendash}L)_{\tmop{no}} \Leftrightarrow (x, a_{| x |}) \in
    L'_{\tmop{no}} & \Leftrightarrow & (x, a_{| x |}) \in (\tmop{co}
    {\textendash}L')_{\tmop{yes}} \Leftrightarrow x \in (\tmop{co}
    {\textendash}L')_{\tmop{yes}} / a\\
    x \in L_{\tmop{no}} \Leftrightarrow x \in (\tmop{co}
    {\textendash}L)_{\tmop{yes}} \Leftrightarrow (x, a_{| x |}) \in
    L'_{\tmop{yes}} & \Leftrightarrow & (x, a_{| x |}) \in (\tmop{co}
    {\textendash}L')_{\tmop{no}} \Leftrightarrow x \in (\tmop{co}
    {\textendash}L')_{\tmop{no}} / a
  \end{eqnarray*}
  so $L = (\tmop{co} {\textendash}L') / a \in (\tmop{co} {\textendash}K) /
  \tmop{mpoly}$, which establishes $(\tmop{co} {\textendash}K) / \tmop{mpoly}
  \supseteq \tmop{co} {\textendash} (K / \tmop{mpoly})$. Choosing $(\tmop{co}
  {\textendash}L') / a$ to be an arbitrary problem in $(\tmop{co}
  {\textendash}K) / \tmop{mpoly}$ and following the argument in reverse
  establishes \ $(\tmop{co} {\textendash}K) / \tmop{mpoly} \subseteq \tmop{co}
  {\textendash} (K / \tmop{mpoly})$.
\end{proof}

\subsection{Collapse theorem for $\tmop{QCPH} / \tmop{mpoly}$}

In this section we will prove Theorems \ref{thm:advice-collapse} and \ref{thm:non-advice-inclusion}. The following is our key lemma:

\begin{lemma}
  \label{lem:advice-collapse}For any $i \in \mathbb{Z}^+$:
  
  (i) $\tmop{QC} \Pi_i \subseteq \tmop{QC} \Sigma_i / \tmop{mpoly}$ implies
  $\tmop{QC} \Sigma_{i + 1} / \tmop{mpoly} = \tmop{QC} \Sigma_i /
  \tmop{mpoly}$
  
  (ii) $\tmop{QC} \Sigma_i \subseteq \tmop{QC} \Pi_i / \tmop{mpoly}$ implies
  $\tmop{QC} \Pi_{i + 1} / \tmop{mpoly} = \tmop{QC} \Pi_i / \tmop{mpoly}$
\end{lemma}

\begin{proof}
  Below we prove part (i) of the theorem. Part (ii) follows from a similar
  argument.
  
  It suffices to show that $\tmop{QC} \Sigma_{i + 1} / \tmop{mpoly} \subseteq
  \tmop{QC} \Sigma_i / \tmop{mpoly}$
  
  Let $L \in \tmop{QC} \Sigma_{i + 1} / \tmop{mpoly}$ and let $V$ be a uniform
  family of verificaiton circuits for $L$ with advice $a \in \tmop{poly}$.
  
  By definition,
  \begin{eqnarray*}
    L_{\tmop{yes}} & = & \{ x| \exists y_1 \forall y_2 \exists y_3 \ldots Q_{i
    + 1} y_{i + 1} \Pr [V (a_{| x |}, x, y_1, y_2, \ldots, y_{i + 1}) = 1]
    \geqslant 2 / 3 \}\\
    L_{\tmop{no}} & = & \{ x| \forall y_1 \exists y_2 \forall y_3 \ldots
    \overline{Q}_{i + 1} y_{i + 1} \Pr [V (a_{| x |}, x, y_1, y_2, \ldots,
    y_{i + 1}) = 1] \leqslant 1 / 3 \}
  \end{eqnarray*}
  Part (i) follows from showing that $L \in \tmop{QC} \Sigma_i /
  \tmop{mpoly}$.
  
  Define a new promise problem $L'$:
  \begin{eqnarray*}
    L'_{\tmop{yes}} & = & \{ (b, x, y_1) | \forall y_2 \exists y_3 \ldots Q_{i
    + 1} y_{i + 1} \Pr [V (b, x, y_1, y_2, \ldots, y_{i + 1}) = 1] \geqslant 2
    / 3 \}\\
    L'_{\tmop{no}} & = & \{ (b, x, y_1) | \exists y_2 \forall y_3 \ldots
    \overline{Q}_{i + 1} y_{i + 1} \Pr [V (b, x, y_1, y_2, \ldots, y_{i + 1})
    = 1] \leqslant 1 / 3 \}
  \end{eqnarray*}
  Note that since $V$ is a verifier for $L'$ and the same $V$ with $a$ is a
  verifier for $L$
  \begin{eqnarray*}
    x \in L_{\tmop{yes}} & \Leftrightarrow & \exists y_1 : (a_{| x |}, x, y_1)
    \in L'_{\tmop{yes}}\\
    x \in L_{\tmop{no}} & \Leftrightarrow & \forall y_1 : (a_{| x |}, x, y_1)
    \in L'_{\tmop{no}} .
  \end{eqnarray*}
  $L'$ is clearly a valid promise problem in $\tmop{QC} \Pi_i$. Applying the
  hypothesis to $L'$ gives us that $L' \in \tmop{QC} \Sigma_i / \tmop{mpoly}$,
  which means that there is a language $L'' \in \tmop{QC} \Sigma_i /
  \tmop{mpoly}$ with uniform verifier circuit $V''$ and advice $a'' \in
  \tmop{poly}$ such that $L'_{\tmop{yes}} \subseteq L''_{\tmop{yes}}$ and
  $L'_{\tmop{no}} \subseteq L''_{\tmop{no}}$. In other words:
  \begin{eqnarray*}
    (b, x, y_1) \in L_{\tmop{yes}}' & \Rightarrow & (b, x, y_1) \in
    L_{\tmop{yes}}''\\
    (b, x, y_1) \in L_{\tmop{no}}' & \Rightarrow & (b, x, y_1) \in
    L_{\tmop{no}}'' .
  \end{eqnarray*}
  Combining the previous two implications gives us that
  \begin{eqnarray*}
    x \in L_{\tmop{yes}} & \Rightarrow & \exists y_1 : (a_{| x |}, x, y_1) \in
    L_{\tmop{yes}}''\\
    x \in L_{\tmop{no}} & \Rightarrow & \forall y_1 : (a_{| x |}, x, y_1) \in
    L_{\tmop{no}}'',
  \end{eqnarray*}
  so
  \begin{eqnarray*}
    x \in L_{\tmop{yes}} & \Rightarrow & \exists y_1 : \exists y_2 \forall y_3
    \ldots Q_i'' y_{i + 1} \Pr [V'' (a''_{| (a_{| x |}, x, y_1) |}, a_{| x |},
    x, y_1, y_2, \ldots, y_{i + 1}) = 1] \geqslant 2 / 3\\
    x \in L_{\tmop{no}} & \Rightarrow & \forall y_1 : \forall y_2 \exists y_3
    \ldots \overline{Q}_i'' y_{i + 1} \Pr [V'' (a''_{| (a_{| x |}, x, y_1) |},
    a_{| x |}, x, y_1, y_2, \ldots, y_{i + 1}) = 1] \leqslant 1 / 3 .
  \end{eqnarray*}
  From this it seems (from the perspective of deciding if $x \in L$) that it
  is unecessary to include $y_1$ and $a_{| x |}$ in $L''$ since can we combine
  $y_1$ and $y_2$ into a single witness and add $a_{| x |}$ to the advice;
  however attempting to use $V''$ without $y_1$ and and $a_{| x |}$ in the
  input poses the following technical problem: $V''$ would recieve the advice
  for the wrong input length $a''_{| x |}$ instead of $a''_{| (a_{| x |}, x,
  y_1) |}$.
  
  The solution is to define a new advice function $a'''$ that combines $a$
  and all possible $a''$:
  \[ a'''_{| x |} \assign (a_{| x |}, a''_{| (a_{| x |}, 1^{| x |}, 1^1) |},
     a''_{| (a_{| x |}, 1^{| x |}, 1^2) |}, \ldots, a''_{| (a_{| x |}, 1^{| x
     |}, 1^{q (| x |)}) |}), \]
  where $q$ is the polynomial upper limit on the size of $y_1$ required by the
  definition of $\tmop{QC} \Sigma_{i + 1} / \tmop{mpoly}$. $a'''$ is a
  poly-sized advice function since $a$ and $a''$ are poly-sized adviced
  functions and $q$ is a polynomial. It is trivial to design a poly-sized
  uniform circuit family which invokes $V''$ with the correct length advice
  when given $a'''$:
  \[ V''' (a'''_{| x |}, x, y_1, y_2, \ldots, y_{i + 1}) \assign V'' (a''_{|
     (a_{| x |}, x, y_1) |}, a_{| x |}, x, y_1, y_2, \ldots, y_{i + 1}) . \]
  By the definitions of $a'''$ and $V'''$ we have that
  \begin{eqnarray*}
    x \in L_{\tmop{yes}} & \Rightarrow & \exists (y_1, y_2) \forall y_3 \ldots
    Q_i'' y_{i + 1} \Pr [V''' (a_{| x |}''', x, y_1, y_2, \ldots, y_{i + 1}) =
    1] \geqslant 2 / 3\\
    x \in L_{\tmop{no}} & \Rightarrow & \forall (y_1, y_2) \exists y_3 \ldots
    \overline{Q}_i'' y_{i + 1} \Pr [V''' (a_{| x |}''', x, y_1, y_2, \ldots,
    y_{i + 1}) = 1] \leqslant 1 / 3,
  \end{eqnarray*}
  from which it follows that $L \in \tmop{QC} \Sigma_i / \tmop{mpoly}$.
\end{proof}

With a little work, we can derive the following corollaries from our key lemma. 

\begin{corollary}
  \label{cor:advice-collapse}The following statements are equivalent:
  
  1. $\tmop{QC} \Pi_i \subseteq \tmop{QC} \Sigma_i / \tmop{mpoly}$
  
  2. $\tmop{QC} \Sigma_i \subseteq \tmop{QC} \Pi_i / \tmop{mpoly}$
  
  3. $\tmop{QC} \Sigma_i / \tmop{mpoly} = \tmop{QC} \Pi_i / \tmop{mpoly}$
  
  4. $\tmop{QCPH} / \tmop{mpoly} = \tmop{QC} \Sigma_i / \tmop{mpoly}$
  
  5. $\tmop{QCPH} / \tmop{mpoly} = \tmop{QC} \Pi_i / \tmop{mpoly}$.
\end{corollary}

\begin{proof}
  Assume the first statement; then by Lemma \ref{lem:advice-collapse}
  $\tmop{QC} \Sigma_{i + 1} / \tmop{mpoly} = \tmop{QC} \Sigma_i /
  \tmop{mpoly}$, so $\tmop{QC} \Pi_i / \tmop{mpoly} \subseteq \tmop{QC}
  \Sigma_i / \tmop{mpoly}$ (by the definition of $\tmop{QC} \Sigma_{i + 1} /
  \tmop{mpoly}$). Then by Lemma \ref{co-mpoly-commute} $\tmop{QC} \Pi_i /
  \tmop{mpoly} \subseteq \tmop{co} {\textendash} (\tmop{QC} \Pi_i /
  \tmop{mpoly})$ which implies $\tmop{QC} \Pi_i / \tmop{mpoly} = \tmop{co}
  {\textendash} (\tmop{QC} \Pi_i / \tmop{mpoly})$ (the third statement). The
  second statement implies the third by a similar argument. The third
  statement trivially implies the first and second statements; and thus the
  first three statements are equivalent.
  
  The fourth and fifth statements trivially imply the first and second
  statements respectively. The third statement implies $\tmop{QC} \Sigma_i /
  \tmop{mpoly} = \tmop{QC} \Pi_i / \tmop{mpoly} = \tmop{QC} \Sigma_{i + 1} /
  \tmop{mpoly} = \tmop{QC} \Pi_{i + 1} / \tmop{mpoly}$ by applying Lemma
  \ref{lem:advice-collapse}; recursively applying Lemma
  \ref{lem:advice-collapse} implies that $\tmop{QC} \Sigma_i / \tmop{mpoly} =
  \tmop{QC} \Pi_i / \tmop{mpoly} = \tmop{QCPH} / \tmop{mpoly}$. Thus all five
  statements are equivalent.
\end{proof}

Theorems \ref{thm:advice-collapse} and \ref{thm:non-advice-inclusion} follow from Corollary \ref{cor:advice-collapse}, which completes our work for this section.

\subsection{$\tmop{QCPH} / \tmop{mpoly}$ collapse implies $\tmop{QCPH}$
collapse}

In this section we will prove Theorem \ref{thm:collapse-collapse}.

Our proof of Theorem \ref{thm:collapse-collapse} uses a polynomial recursive
structure (also known as a self-reducibility structure) for a $\tmop{QC}
\Pi_i$-complete problem. We will use Trivial complete problem for $\tmop{QC}
\Pi_i$, defined below.

\begin{definition}
  ($\tmop{TCPQC} \Pi_i$: Trivial complete problem for $\tmop{QC} \Pi_i$)
  $\tmop{TCPQC} \Pi_i = ((\tmop{TCPQC} \Pi_i)_{\tmop{yes}}, (\tmop{TCPQC}
  \Pi_i)_{\tmop{no}})$ such that
  \begin{eqnarray*}
    V, x \in (\tmop{TCPQC} \Pi_i)_{\tmop{yes}} & \Leftrightarrow & \forall y_1
    \exists y_2 \forall y_3 \ldots \overline{Q}_i y_i \Pr [V (x, y_1, y_2,
    \ldots, y_i) = 1] \geqslant 2 / 3\\
    V, x \in (\tmop{TCPQC} \Pi_i)_{\tmop{no}} & \Leftrightarrow & \exists y_1
    \forall y_2 \exists y_3 \ldots Q_i y_i \Pr [V (x, y_1, y_2, \ldots, y_i) =
    1] \leqslant 1 / 3,
  \end{eqnarray*}
  where the $y_i$ are polynomial-size classical bit strings and each instance
  of $\tmop{TCPQC} \Pi_i$ consists of a quantum circuit $V$ and an instance
  $x$.
\end{definition}

Next we will define a recursive structure (c.f. \cite{Yap83}, \cite{KL80} self-reducible structure).

\begin{definition}
  (Recursive Structure) A recursive structure for a language $L =
  (L_{\tmop{yes}}, L_{\tmop{no}})$ is a tuple of functions $(B, F_1, \ldots,
  F_r, G_1, \ldots, G_s)$ which satisfy the following conditions:
  \begin{enumerateroman}
    \item $\forall (x \in L_{\tmop{yes}} \cup L_{\tmop{no}}) [D_L (x) = B (F_1
    (x), \ldots, F_r (x), D_L (G_1 (x)), \ldots, D_L (G_s (x)))]$ where $D_L$
    is the decision function for $L$ ($D_L (x) = \tmop{YES} \Leftrightarrow x
    \in L_{\tmop{yes}}$ and $D_L (x) = \tmop{NO} \Leftrightarrow
    L_{\tmop{no}}$).
    
    \item $\{ G_1, \ldots, G_s \}$ induces a well-founded strict partial
    ordering $\prec$ on $L_{\tmop{yes}} \cup L_{\tmop{no}}$; this means that
    for all $x \in L_{\tmop{yes}} \cup L_{\tmop{no}}$ and $k \in [s]$, $G_k
    (x) \in L_{\tmop{yes}} \cup L_{\tmop{no}}$ and either $G_k (x) \prec x$ or
    $G_k (x) = x$.
    
    \item We call $x$ a $\prec$-minimal element iff $G_k (x) = x$ for all
    $G_k$. If $x$ is $\prec$-minimal, then $B$ is only senstive to $F_1,
    \ldots, F_r$. In other words,
    \[ B (F_1 (x), \ldots, F_r (x), y_1, \ldots, y_s) = B (F_1 (x), \ldots,
       F_r (x), y_1', \ldots, y_s') \]
    for all $y_1, \ldots, y_s$ and $y'_1, \ldots, y'_s$.
  \end{enumerateroman}
\end{definition}

\begin{definition}
  (Polynomial Recursive Structure) A Polynomial Recursive Structure is a
  recursive structure $(B, F_1, \ldots, F_r, G_1, \ldots, G_s)$ where every
  function in $(B, F_1, \ldots, F_r, G_1, \ldots, G_s)$ is computable in
  quantum polynomial time. 
\end{definition}

Now we can describe a polynomial recursive structure for $\tmop{TCPQC} \Pi_i$.
Let $A_i$ denote some fixed reasonable encoding of the sentence
\[ \forall y_1 \exists y_2 \forall y_3 \ldots Q_i y_i V (x, y_1, y_2, \ldots,
   y_i) \]
as a classical bit string, where $V$ is a polynomial-size quantum circuit with
input registers $x, y_1, y_2, \ldots, y_i$ and $x$ is a classical bit string.
Define $E_i$ similarly, except with alternating quantifiers starting with
$\exists$. Clearly $A_i$ is also a reasonable way to encode a $\tmop{TCPQC}
\Pi_i$ instance. We describe a polynomial recursive structure for
$\tmop{TCPQC} \Pi_i$ encoded in this way. At a high level, the way the
structure works is that it tries both values for one bit of $y_i$ at a time
and runs $V$ if all the bits of $y_i$ have already been chosen. The structure
is $(B, F_1, F_2, G_0, G_1)$, where
\[ B (f_1, f_2, g_0, g_1) = \left\{\begin{array}{ll}
     f_2 & \tmop{if} f_1 = 0\\
     g_0 \vee g_1 & \tmop{if} f_1 = 1 \wedge f_2 = 1\\
     g_0 \wedge g_1 & \tmop{if} f_1 = 1 \wedge f_2 = 0
   \end{array}\right. \]
; $F_1 (a) = 1$ iff a encodes a sentence where $V$ is non-constant ($V$ has at
least one variable amongst its input registers), $F_2 (a) = 1$ iff $V$ is
constant (all the bits in its input registers have been given a value in $\{
0, 1 \}$) and evaluates to YES with probability $\geqslant 2 / 3$ or the
leftmost variable $y_i$ that occurs in an input register of $V$ is bounded by
$\exists$; \ $G_0 (a) = G_1 (a) = a$ if $F_1 (a) = 0$, otherwise $G_b (a)$
fixes the leftmost variable bit of the input register for $y_i$ in $V$ to $b$
where $y_i$ is the leftmost variable such that not all the bits in the
register for $y_i$ in $V$ have been fixed yet. Furthermore, we us a suitable
encoding such that $G_b$ is length preserving, so $| G_b (a) | = | a |$. It's
clear that $(B, F_1, F_2, G_0, G_1)$ can be computed in polynomial quantum
time and that all the functions can be computed deterministically except $F_2$
which can have exponentially small error probability by simulating $V$ a
polynomial number of time and taking the majority. It's also clear that $(B,
F_1, F_2, G_0, G_1)$ decides $\tmop{TCPQC} \Pi_i$, gives a partial ordering,
and has minimal elements since each application of $G_b$ decreases the number
of variable bits until there are none. By definition, $B$ is only sensitive to
$F_1$ and $F_2$ on the minimal elements where all the bits are fixed.

We are now ready to prove Theorem \ref{thm:collapse-collapse}. The high level
idea is to decide an instance of an advice problem with a polynomial recursive
structure in $\tmop{QC} \Sigma_{i'}$ by using an existential quantifier to get
some candidate advice and then using a universal quantifier and the recursive
structure to check that the candidate advice gives correct results for the
instance.

\begin{proof}
  Suppose that $\tmop{QC} \Sigma_i / \tmop{mpoly} = \tmop{QC} \Pi_i /
  \tmop{mpoly}$. It suffices to show that \ $\tmop{QC} \Pi_{i + 2} \subseteq
  \tmop{QC} \Sigma_{i + 2}$, by invoking the collapse theorem for QCPH
  (Theorem \ref{thm:qcph-collapse}).
  
  By Corollary \ref{cor:advice-collapse}, $\tmop{QCPH} / \tmop{mpoly}$
  collapses to the $i^{\tmop{th}}$ level, so $\tmop{QC} \Pi_{i + 2} \subseteq
  \tmop{QC} \Pi_{i + 2} / \tmop{mpoly} \subseteq \tmop{QC} \Pi_i /
  \tmop{mpoly}$, so $\tmop{TCPQC} \Pi_{i + 2} \in \tmop{QC} \Pi_i /
  \tmop{mpoly}$. Then there is some $P_i \in \tmop{QC} \Pi_i$ and advice
  function $a$ such that $x \in \tmop{TCPQC} \Pi_{i + 2} \Leftrightarrow (x,
  a_{| x |}) \in P_i$ where $x$ is encoded in the form $A_{i + 2}$. We define
  the following predicate $W (x, w)$:
  \[ \forall y [| y | = | x | \Rightarrow D_{P_i} (y, w) = B (F_1 (y), F_2
     (y), D_{P_i} (G_0 (y), w), D_{P_i} (G_1 (y), w))], \]
  where $D_{P_i}$ is a decision function for $P_i$ and $(B, F_1, F_2, G_0,
  G_1)$ is the polynomial recursive structure for $\tmop{TCPQC} \Pi_{i + 2}$.
  Since we may be computing $D_{P_i}$ using a randomized quantum algorithm we
  cannot directly evaluate the rightmost equality of $W (x, w)$, so we will
  use the following procedure instead:
  \begin{enumerateroman}
    \item Determine the number $n_i$ of input bits that are not yet fixed in
    $y$.
    
    \item Fix some $\tmop{QC} \Pi_i$ algorithm $\mathcal{A}'$ for $D_{P_i}$.
    
    \item Let the algorithm $\mathcal{A}$ for $D_{P_i} (y, w)$ return
    $\mathcal{A}' (y, w)$ when the number of input bits not yet fixed in $y$
    is $\leqslant n_i - 1$ and return $B \left( F_1 (y), F_2 (y) {,
    \mathcal{A}'}^7 (G_0 (y), w) {, \mathcal{A}'}^7 (G_1 (y), w) \right)$ when
    the number of bits not yet fixed in $y$ is $n_i$ and $w = a_{| x |}$,
    where ${\mathcal{A}'}^k$ represents running $\mathcal{A}'$ $k$ times and
    taking the majority. $\mathcal{A}$ is a valid $\tmop{QC} \Pi_i$ algorithm
    for $D_{P_i}$ since $\mathcal{A}'$ is a valid $\tmop{QC} \Pi_i$ algorithm
    for $D_{P_i}$, the probability that two runs of ${\mathcal{A}'}^7$ both
    make a correct decision is $\approx 0.68 > 2 / 3$, and the recursive
    structure is correct and deterministic. Note that we have not specified
    what $\mathcal{A}$ returns when $w \neq a_{| x |}$.
    
    \item Use $\mathcal{A}$ to compute the left hand side $\tmop{poly} (n)$
    times. Categorize the fraction of accept results into one of the following
    three groups: less than $\frac{1}{3} + \frac{1}{h (n)}$, more than
    $\frac{2}{3} - \frac{1}{h (n)}$, and other, for some $h (n) \in
    \tmop{poly} (n)$.
    
    \item Use $\mathcal{A}^7$ to compute the right hand side $\tmop{poly} (n)$
    times and choose from the same categories.
    
    \item The equality is satisfied if and only if the left and right hand
    sides are in the same category.
  \end{enumerateroman}
  Let $W' (x, w)$ denote $W (x, w)$ computed using the above procedure the
  evaluate the rightmost equality. $\Pr [W' (x, a_{| x |}) \tmop{accepts}]
  \geqslant \frac{2}{3}$ when restricted to $y \in (\tmop{TCPQC}
  \Pi_i)_{\tmop{yes}} \cup (\tmop{TCPQC} \Pi_i)_{\tmop{no}}$ by the correctnes
  of $D_{P_i}$ and the recursive structure and $\Pr [W' (x, a_{| x |})
  \tmop{accepts}] \geqslant \frac{2}{3}$ for the remaining values of $y$
  because both sides of the rightmost equality are drawn from the same
  distribution because they are computed the same way. We claim that the
  following sentence on $x$ decides $\tmop{TCPQC} \Pi_{i + 2}$:
  \begin{equation}
    \label{eq:qcpi2-complete} \exists w ((x, w) \in P_i \wedge W' (x, w)) .
  \end{equation}
  First we will establish that $x \in (\tmop{TCPQC} \Pi_{i + 2})_{\tmop{yes}}$
  implies that sentence \ref{eq:qcpi2-complete} is true: choose $w = a_{| x
  |}$, then $\Pr [W' (x, a_{| x |}) \tmop{accepts}] \geqslant \frac{2}{3}$ and
  $(x, w) \in P_i$. Next we establish that the truth of sentence
  \ref{eq:qcpi2-complete} implies that $x \in (\tmop{TCPQC} \Pi_{i +
  2})_{\tmop{yes}}$. Fix a value of $w$ which makes the sentence true (with
  bounded probability). We show that $(z, w) \in P_i \Rightarrow z \in
  \tmop{TCPQC} \Pi_{i + 2}$ by induction on $| x |$-length strings $z$ that
  encode the same instance as $x$ up to fixed bits in the input registers of
  the circuit: assume that $(G_b (z), w) \in P_i \Rightarrow G_b (z) \in
  \tmop{TCPQC} \Pi_{i + 2}$ for $b \in \{ 0, 1 \}$, then $D_{P_i} (z, w) = B
  (F_1 (z), F_2 (z), (G_0 (z) \in \tmop{TCPQC} \Pi_{i + 2}), (G_1 (z) \in
  \tmop{TCPQC} \Pi_{i + 2}))$ so $(D_{P_i} (z, w) = 1) \Rightarrow z \in
  \tmop{TCPQC} \Pi_{i + 2}$ by the correctness of the recursive structure. We
  use $\prec$-minimal $z$ \ as our base case; then $B (F_1 (z), F_2 (z), (z,
  w), (z, w)) = D_{P_i} (z, a_{| x |})$ because $B$ is not sensitive to $w$.
  
  Finally we show that the truth of sentence \ref{eq:qcpi2-complete} can be
  decided in $\tmop{QC} \Sigma_{i + 2}$, which completes the proof. $(x, w)
  \in P_i$ can be correctly decided in $\tmop{QC} \Pi_i$, computing $W' (x,
  w)$ is in $\tmop{QC} \Pi_{i + 1}$ because $\mathcal{A}$ is in $\tmop{QC}
  \Pi_i$ and the remaining operations can be done in polynomial time, and both
  success probabilities can be boosted in polynomial time so that the
  conjunction is correct with bounded probability. Thus deciding if sentence
  \ref{eq:qcpi2-complete} is true is in $\tmop{QC} \Sigma_{i + 2}$.
\end{proof}

\bibliographystyle{alpha}
\bibliography{references}

\end{document}